\documentclass[10pt]{article}
\usepackage[letterpaper]{geometry}
\usepackage{hicss51}
\usepackage{times}
\usepackage[none]{hyphenat}
\usepackage{url}
\usepackage{latexsym}
\usepackage[draft]{minted}
\usepackage{indentfirst}
\usepackage{graphicx}
\graphicspath{{images/}}

\usepackage{graphicx}
\usepackage{color}
\usepackage{nicefrac}
\usepackage{pgf}
\usepackage{algorithm}
\usepackage[noend]{algpseudocode}

\usepackage{todonotes}
\usepackage{framed}

\usepackage{amsmath}
\usepackage{amsfonts}
\usepackage{amsthm}
\usepackage{amssymb}

\newtheorem{definition}{Definition}
\newtheorem{example}{Example}
\newtheorem{corollary}{Corollary}
\newtheorem{lemma}{Lemma}
\newtheorem{theorem}{Theorem}
\newtheorem{remark}{Remark}

\newcommand{\Diff}{\textrm{Diff}}
\newcommand{\Ind}{\textrm{Ind}}

\newcommand{\cost}{\textrm{cost}}

\newcommand{\dom}{\vartriangleright}
\newcommand{\dpath}{\vartriangleright\vartriangleright\vartriangleright}
\usepackage{latexsym}
\newcommand{\gpath}{\leadsto}

\newcommand{\splittt}{\textrm{split}}
\newcommand{\pref}{\succ}

\newcommand{\mypara}[1]{\smallskip\noindent\textbf{#1.}\quad}

\title{A Participatory Democratic Budgeting Algorithm}

\author{Ehud Shapiro \\
  Weizmann Institute of Science \\
  {\underline{ ehud.shapiro@weizmann.ac.il}} \\\And
  Nimrod Talmon \\
  Ben-Gurion University \\
  {\underline{ talmonn@bgu.ac.il} }}

\date{}

\begin{document}

\pagestyle{plain}

\maketitle

\begin{abstract}
The budget is the key means for effecting policy in democracies, yet its preparation is typically an excluding, opaque, and arcane process. We aim to rectify this by providing for the democratic creation of complete budgets --- for cooperatives, cities, or states.  Such budgets are typically (i) prepared, discussed, and voted upon by comparing and contrasting with last-year's budget; (ii) quantitative, in that items appear in quantities with potentially varying costs; and (iii) hierarchical, reflecting the organization's structure.   Our process can be used by a budget committee, the legislature or the electorate at large.

We allow great flexibility in vote elicitation, from perturbing last-year's budget to a complete ranked budget proposal.  We present a polynomial-time algorithm which takes such votes, last-year's budget, and a budget limit as input and produces a budget that is provably ``democratically optimal'' (Condorcet-consistent), in that no proposed change to it has majority support among the votes. 
\end{abstract}

\section{Introduction}

Even in the best democracies the budgeting process is somewhat of a mystery. Various people and bodies are involved, officially and unofficially, in the limelight and behind the scenes.
Arguments are made, pressures are mounted, and at the end of the process a budget proposal is placed in front of the governing body for approval.
While the governing body may question the proposed budget and request amendments, there is presently no feasible method for turning the individual preferences of the members of the governing body into a budget that faithfully reflects their joint democratic will. Here we provide such a method, which can be used by a budget committee, the parliament, or the citizens at large.

Participatory budgeting (PB)~\cite{cabannes2004participatory} 
 offers ways for a society to participate in constructing a budget.
Initiated by Brazil workers' party~\cite{wainwright2003making},
various societies have adopted similar ideas~\cite{ganuza2012power},
where usually only a small fraction of a municipality's annual budget is decided directly by the residents.
After a deliberative phase,
participants vote on projects such as schools, bike roads, etc.,
usually by some variants of $k$-Approval~\cite{kilgour2010approval},
where each resident specifies $k$ projects of highest priority
(other elicitation methods are mentioned below).

Current methods of participatory budgeting may be limited not only in their application, but also in their foundations: 
  (1) While the previous budget is typically the starting point of practical budgeting processes, current methods ignore it;
  (2) The elicitation methods are usually limited to approvals or simple pairwise comparisons;
  (3) The methods typically cannot handle budgeting of quantitative items (e.g., how many new school buses to buy?) with all their complexities (e.g., decreasing marginal costs);
  and
  (4) Current methods cannot deal with hierarchical budgeting.

Here we aim to address these issues.
As our method supports the democratic conduct of budget committees and parliaments, we term our approach \emph{democratic budgeting}.
We consider one of the foundations of social choice,
namely the Condorcet principle,
and generalize it to our setting:
  In single winner elections, the Condorcet approach involves selecting a winner which dominates all others by a voter majority;
a natural way to generalize it to budgets is to select a budget (i.e., a subset of the proposed items, the cost of which is within the budget limit)
which dominates all other budgets.

We describe a polynomial-time budgeting algorithm that computes a Condorcet-winning budget, whenever such a budget exists; if no such budget exists, it finds a member of the \emph{Smith set}, defined below.
We first prove the correctness of our algorithm on a simple model,
which is close to the standard model of participatory budgeting
(voters specify rankings, items are not quantitative), but takes the previous budget into consideration;
then, we enhance our model as discussed above and describe how our algorithm generalizes naturally and remains computationally efficient.

Our democratic budgeting scenario is natural and captures more situations than the usual model;
it is powerful and allows voters to specify both very simple preferences (e.g., approvals) and very complex ones (e.g., partial orders);
budgets computed by our algorithm satisfy the voters' wishes,
and cannot be argued against.
Our method allows for the construction of hierarchical budgets,
and thus can be applied to situations where a ``high-level'' budget
(say, a city budget) consists of several ``low-level'' budgets
(say, one for education, another for transportation, etc.).

\subsection{Related Work}\label{section:related work}

Our model differs from existing models of PB:
(1) We directly incorporate last year's budget, in accordance with Reality-aware Social Choice~\cite{realsoc};
(2) We employ a powerful and flexible elicitation method, accommodating partial orders (contrasted to Approval voting~\cite{cabannes2004participatory,aussieone}, Knapsack voting~\cite{goel2015knapsack}, Value voting, Value-for-money voting, and Threshold voting~\cite{benade2017preference});
(3) We allow quantifiable and indivisible items of various costs, allowing to express, e.g., decreasing marginal costs of budgeting certain amounts of the same item;
and
(4) We provide for hierarchical budget construction.

Within computational social choice~\cite{moulin2016handbook},
multiwinner elections~\cite{chaptermw} are receiving increased attention.
We generalize multiwinner elections 
as (1) We allow partial orders (compare to, e.g., \cite{csrhierarchy}; see also, e.g.,~\cite{xia2011determining})
and (2) We allow different costs (compare to, e.g., \cite{chaptermw}; see also, e.g., \cite{lu2011budgeted}).
Our generalization of the Condorcet criterion generalizes that of Fishburn~\cite{fishburn1981majority,fishburn1981analysis}
(studied further in~\cite{darmann2013hard} and~\cite{sekar2017condorcet};
different, element-wise concepts exist~\cite{kamwa2017stable,aziz2017condorcet}).
We concentrate on the more abstract minmax set extension~\cite{barbera2004ranking,aziz2016computing};
other set extensions are studied~\cite{darmann2013hard,aziz2016computing}).

Our efficient algorithm iteratively uses Schwartz's tournament solution. In this we mention the work of Bouysso~\cite{Bouy04a} which study such algorithms abstractly.

\section{Preliminaries}\label{section:preliminaries}

We provide preliminaries regarding rankings and tournaments.
For an integer $n \in \mathbb{N}$, we denote the set $\{1, \ldots, n\}$ by $[n]$.
Given a set $A$, a \emph{partial order} $\pref$ on $A$ is a reflexive, antisymmetric, and transitive relation.
A \emph{ranking} is a total order.
A \emph{weak ranking} is a total order with ties, and corresponds to an \emph{ordered partition}:
  Given a set $A$, an ordered partition $\pref$ is a partition of $A$ into linearly ordered and disjoint sets, referred to as \emph{components}, whose union is $A$.
E.g., given $A = \{1, 2, 3, 4, 5\}$,
the ordered partition $\{2\} \pref \{3, 4\} \pref \{1, 5\}$,
whose first, second, and third components are, correspondingly, $\{2\}$, $\{3, 4\}$, and $\{1, 5\}$
defines a weak order where, e.g., $2$ is preferred to $4$ while $3$ and $4$ tie.
We interchange \emph{ordered partitions}, \emph{weak orders}, and \emph{weak rankings}.
A \emph{tournament} is a directed graph.
A \emph{tournament solution} is a function $f$ which takes a tournament as an input and selects a nonempty subset of its vertices as output; i.e., $f : \mathcal{T}(A) \to 2^{A} \setminus \emptyset$, where $\mathcal{T}(A)$ is the set of all tournaments over $A$.
We use the following two tournament solutions.

\begin{definition}
A \emph{Schwartz component} of a tournament $T = (A, E)$ is
a minimal set $X \subseteq A$ of vertices of $T$
such that
for any $b \in A \setminus X$ there is no $a \in X$ such that $(b, a) \in E$.
The \emph{Schwartz set} is the union of all Schwartz components.
\end{definition}

\begin{definition}
The \emph{Smith set} of a tournament $T = (A, E)$ is
the minimal set $X \subseteq A$ of vertices of $T$
such that
for each $a \in X$ and  each $b \notin X$,
$(a, b) \in E$ holds.
\end{definition}

\section{The Basic Model}\label{section:basic model}

We describe the basic model,
where voters express preferences over distinct items,
and, given the previous budget,
we shall select a subset of these items whose total cost respects the predefined budget limit.
Formally,
a \emph{budget proposal} $\mathcal{P}$ is a set of distinct \emph{budget items},
where the \emph{cost} of a budget item $b \in \mathcal{P}$ is denoted by $\cost(b)$.
A \emph{vote profile} $\mathcal{V}$ is a set of \emph{votes},
where each vote $v \in \mathcal{V}$ is a linear order of the budget proposal $\mathcal{P}$, which specifies the preferences of the voter.
A \emph{budget} $B$ is a subset of the budget proposal $\mathcal{P}$ (sometimes called \emph{bundle} in the literature).
Given a \emph{budget limit} $\ell$, a budget $B$ is \emph{feasible} if its cost is within the budget limit, i.e., if $\sum_{b \in B} \cost(b) \leq \ell$. 
Members of $B$ are \emph{budgeted items} while members of $\mathcal{P} \setminus B$ are \emph{unbudgeted items}.
A budget $B$ is \emph{exhaustive} if it is feasible and for no item $b \in \mathcal{P} \setminus B$,
the budget $B \cup \{b\}$ is feasible.

A \emph{budgeting algorithm} gets as input the tuple $(\mathcal{P}, \mathcal{V}, \ell, B')$,
where
$\mathcal{P}$ is a budget proposal,
$\mathcal{V}$ is a vote profile,
$\ell$ is a budget limit,
and $B_{-1} \subseteq \mathcal{P}$ is the previous (e.g., last year's) budget.
Given this input, a budgeting algorithm computes a feasible budget $B \subseteq \mathcal{P}$.
A budgeting algorithm is \emph{exhaustive} if it computes only exhaustive budgets.
We wish that the budget $B$ computed by a budgeting algorithm,
in addition to being feasible and exhaustive,
will also reflect the preferences of the votes of the vote profile $\mathcal{V}$,
and we will use the previous budget $B_{-1}$ as a guiding compass.

\begin{example}
Consider the following simple budgeting scenario.
The budget proposal $\mathcal{P}$ contains the items: $a$, $b$, and $c$.
The cost of $a$ is $1$, the cost of $b$ is $2$, and the cost of $c$ is $4$.
The vote profile $\mathcal{V}$ contains the votes:
  $v_1 : a \pref b \pref c$
  and
  $v_2 : c \pref a \pref b$.
The budget limit is $\ell = 3$.
Then, $\{a\}$, $\{b\}$, and $\{a, b\}$ are all feasible,
while $\{a, b\}$ is the only exhaustive budget.
An exhaustive budgeting algorithm,
given the tuple $(\mathcal{P}, \mathcal{V}, \ell, B')$,
where $B_{-1} = \{a\}$ is the previous budget,
shall output the budget $B = \{a, b\}$.
\end{example}

\subsection{Condorcet-winning Budgets}\label{section:axioms}

We are interested in applying the Condorcet principle to budgets.
In single-winner elections, the Condorcet principle states that if there is a candidate $c$ such that, for each other candidate $c' \neq c$, more voters prefer $c$ to $c'$, then $c$ shall be chosen.

To apply the Condorcet principle to budgets,
a notion of voter preference among budgets shall be derived from voter preference among budget items.  Fishburn et al.~\cite{fishburn1981majority,fishburn1981analysis} generalized the notion of Condorcet winner to multiwinner elections (where the aim is to select a set of unit-cost items); we further generalize it to sets of items of arbitrary costs. To achieve that, we apply the \emph{minmax} set extension~\cite{barbera2004ranking} to our budgeting setting:
  Intuitively, a voter prefers one budget over another if it ranks higher all the items budgeted by the first budget but not the second, compared to all items budgeted by the second budget but not by the first. 

\newcommand{\pos}{\textrm{pos}}

\begin{definition}[Position]\label{def:position}
Let $v : b_1 \pref b_2 \pref \dots \pref b_{|\mathcal{P}|}$ be a vote over a proposal $\mathcal{P}$.
We define \emph{$\pos_v(B)$, the positions of a budget $B \subseteq \mathcal{P}$ in $v$}, to be $\pos_v(B) = \{i: b_i \in B\}$.
\end{definition}

E.g., for $v : b_1 \pref b_2 \pref b_3$, $\pos_v(\{b_1, b_3\}) = \{1, 3\}$.

\begin{definition}[Prefers]\label{def:prefers}
Let $v : b_1 \pref b_2 \pref \dots \pref b_{|\mathcal{P}|}$ be a vote over a proposal $\mathcal{P}$. Let $B \subseteq \mathcal{P}$ and $B' \subseteq \mathcal{P}$ be two budgets.
Then, $v$ \emph{prefers} $B$ over $B'$ if $\max(\pos_v(B \setminus B')) < \min(\pos_v(B' \setminus B))$, where
$\max(\emptyset) = \min(\emptyset) = \infty$.
\end{definition}

\begin{example}
Let $\mathcal{P} = \{b_1, b_2, b_3, b_4\}$ be a proposal, where
each item costs $1$ and let $v : b_1 \pref b_2 \pref b_3 \pref b_4$ be a vote.
Let $B = \{b_1, b_2, b_3\}$ and $B' = \{b_1, b_2, b_4\}$ be two budgets.
Then,
$\pos_v(B \setminus B') = \pos_v(\{b_3\}) = \{3\}$
while
$\pos_v(B' \setminus B) = \pos_v(\{b_4\}) = \{4\}$,
thus,
$3 = \max(\pos_v(B \setminus B')) < \min(\pos_v(B' \setminus B)) = 4$,
and so $v$ prefers $B$ over $B'$. 
\end{example}

Our definition is conservative as it may refrain from judging one budget as preferable to another even in cases when one might be intuitively inclined to make such a judgment. We find that it offers a good balance in being restrictive, and thus making only solid judgments, but powerful enough to provide a practical foundation for the following definitions and algorithm.
Other options to define voter preference are discussed in Section~\ref{section:outlook}.
Next we describe what it means for one budget to dominate another. Note that we use strict rather than relative majority in defining this notion.

\begin{definition}[Dominance]
Let $\mathcal{V}$ be a profile and $B$ and $B'$ budgets over a proposal $\mathcal{P}$.
We say that
(1) $B$ \emph{dominates} $B'$
if the number of votes preferring $B$ over~$B'$
is larger than $\frac{1}{2}|V|$;
(2) $B$ \emph{weakly dominates} $B'$, denoted by $B \dom B'$, if $B$ is not dominated by $B'$.
\end{definition}

The weak dominance relation $\dom$ is reflexive ($B \dom B$ for all $B$) and total ($B \dom B'$ or $B' \dom B$ or both for all $B$ and $B'$) but not transitive (see Example~\ref{example:paradox}). 

\begin{definition}[Condorcet-winning budget]
Given a vote profile $\mathcal{V}$ over a budget proposal $\mathcal{P}$ and a budget limit $\ell$,
a feasible budget $B$ of $\mathcal{P}$ is a \emph{Condorcet-winning budget} if 
$B$ dominates all other feasible budgets.
\end{definition}

A Condorcet-winning budget cannot be argued against,
as there is no majority to support any change to it.
Condorcet-winning budgets are also exhaustive.
Such budgets, however, do not always exist.

\begin{example}\label{example:paradox}
Consider the budget proposal $\mathcal{P} = \{b_1, b_2, b_3\}$,
where the cost of each item is $1$,
and the following vote profile over it:
$v_1 : b_1 \pref b_2 \pref b_3$,
$v_2 : b_2 \pref b_3 \pref b_1$,
$v_3 : b_3 \pref b_1 \pref b_2$.
Following the symmetry of the profile,
we can assume,
without loss of generality,
that for the budget limit $\ell = 1$,
any budget $B$ which is within limit $\ell$ equals to one of the following two options:
$B_0 = \emptyset$;
or $B_1 = \{b_1\}$.
Further,
budget $B_0$ is not exhaustive, since, e.g., $b_1$ can fit within limit, and hence not a Condorcet winner (as Condorcet-winning budgets are exhaustive).
Regarding budget $B_1$, both $v_2$ and $v_3$ prefer the budget $B_3 = \{b_3\}$ over  $B_1 = \{b_1\}$ and hence $B_1$ is not a Condorcet winner.
\end{example}

\subsection{An Efficient Condorcet-consistent Budgeting Algorithm}

A \emph{Condorcet-consistent budgeting algorithm} is a budgeting algorithm that returns a Condorcet-winning budget whenever such exists.
Consider the following naive algorithm:
  First, it creates a directed graph, referred to as the \emph{budgets graph}, which contains a vertex for each feasible budget and an arc from vertex $B$ to vertex $B'$ if $B$ dominates $B'$.
Given this budgets graph,
the algorithm returns the budget corresponding to a vertex with outgoing arcs to all other budgets,
if such a vertex exists.
Indeed, this budgeting algorithm is Condorcet-consistent.
As there are exponentially-many feasible budgets, the naive algorithm described above runs in exponential time.

We describe a polynomial-time Smith-consistent budgeting algorithm (SBA):
  It always returns a budget from the Smith set of the budgets graph (see Preliminaries, Section~\ref{section:preliminaries}). If a Condorcet-winning budget exists, then the Smith set is its singleton, thus our algorithm is also Condorcet-consistent.
A pseudo-code of SBA is in Algorithm~\ref{algorithm:dob} (notice that we use the previous budget $B_{-1}$; we defer a discussion on it to Section~\ref{section:hysteresis}).
it is composed of these procedures:
  A \textbf{Ranking procedure} which,
  given a vote profile $\mathcal{V}$ over a budget proposal $\mathcal{P}$, aggregates the votes in $\mathcal{V}$ into an ordered partition and outputs it;
  and a \textbf{Pruning procedure} which,
  given an ordered partition and a budget limit, outputs a feasible budget.

\begin{algorithm}[t]
  \caption{Smith-consistent budgeting (SBA)}\label{algorithm:dob}
\begin{algorithmic}[1]

\Procedure{Ranking}{}
\State \textbf{input: } a budget proposal $\mathcal{P}$
\State \textbf{input: } a vote profile $\mathcal{V}$
\State $V \gets \emptyset$
\While{$G \neq \emptyset$}
  \State $S \gets \textrm{Schwartz-set}(G)$
  \State\label{line:schwartz} $V \gets V \pref S$; $G \gets G \setminus S$
\EndWhile
\State \textbf{return} $V$
\EndProcedure
\ \\
\Procedure{Pruning}{}
\State \textbf{input: } a budget proposal $\mathcal{P}$
\State \textbf{input: } an ordered partition $V : C_1 \pref \dots \pref C_z$
\State \textbf{input: } a budget limit $\ell$
\State \textbf{input: } the previous budget $B_{-1}$
\State $B \gets \emptyset$; 
\For{$i \in [z]$}
  \State\label{line:greedy} $B_i \gets$ a maximal subset of $C_i$ closest to $B_{-1}$ with $\cost(B \cup B_i) \leq \ell$
  \State $B \gets B \cup B_i$
\EndFor
\State \textbf{return} $B$
\EndProcedure
\ \\
\Procedure{SBA}{}
\State \textbf{input: } a budget proposal $\mathcal{P}$
\State \textbf{input: } a vote profile $\mathcal{V}$
\State \textbf{input: } a budget limit $\ell$
\State \textbf{input: } the previous budget $B_{-1}$
\State $V \gets \textsc{Ranking}(\mathcal{P}, \mathcal{V})$
\State $B \gets \textsc{Pruning}(\mathcal{P}, V, \ell, B_{-1})$
\State \textbf{return} $B$
\EndProcedure

\end{algorithmic}
\end{algorithm}

The following notions play a key role in the description of the algorithm and its proof.

\begin{definition}[Majority graph, weak majority graph]
Let $\mathcal{P}$ be a budget proposal and `$\mathcal{V}$ be a vote profile.  The \emph{majority graph} of $\mathcal{P}$ and $\mathcal{V}$ has a vertex for each item $b \in \mathcal{P}$ and an arc between any two items $b, b' \in \mathcal{P}$ if more than half of the votes rank $b'$ higher than $b$.

Its corresponding \emph{weak majority graph} has, in addition, arcs $(b,b')$ and $(b',b)$ for every $b, b' \in \mathcal{P}$ for which the majority graph has neither arc.
\end{definition}

The \textbf{Ranking procedure} first creates the \emph{majority graph}.
Then, it iteratively construct the ordered partition~$V$, initiated to be empty.  In each iteration, it identifies the Schwartz set (see Section~\ref{section:preliminaries}), adds its vertices as the next component of $V$ and removes them from the majority graph.
When no vertices remain in the majority graph it halts, at which point $V$ is indeed an ordered partition of $\mathcal{P}$.

The \textbf{Pruning procedure}, given the ordered partition $V : C_1 \pref C_2 \pref \dots \pref C_z$ computed by the ranking procedure,
gradually populates the set of budgeted items $B$.
It iterates over the components of $V$, where the $i^{th}$ iteration considers the $i^{th}$ component $C_i$.
It chooses a maximal subset $B_i$ of $C_i$ such that the cost of $B \cup B_i$ remains within limit;
if several such subsets exist,
it chooses one which is the closest to the previous budget $B_{-1}$,
where closeness between two budgets is measured by the total cost of the symmetric difference between them.
In particular, observe that if $\cost(B \cup C_i) \leq \ell$ then $B_i = C_i$.
If this is not the case, then some items of $C_i$ will remain unbudgeted, but budgeting any of them would cause $B$ to go over the limit.
After considering all components of $V$, it outputs the budget $B$.

\subsection{SBA is Condorcet-consistent}

Here we prove that SBA is Smith-consistent, and thus in particular Condorcet-consistent (see Section~\ref{section:preliminaries}).
We refer to a budget as a \emph{Smith budget} if it is in the Smith set and show that our algorithm computes only Smith budgets.
To do so we introduce the following notions.

\begin{definition}[Weak domination path]
There is a \emph{weak domination path} (\emph{path} for short) in the budgets graph, from a budget $B$ to a budget $B'$,
denoted by $B \dpath B'$,
if there is a sequence of budgets $(B_1, B_2, \ldots, B_k)$, $k \geq 2$, such that $B = B_1$, $B_k = B'$, and $B_i \dom B_{i + 1}$ for each $i \in [k - 1]$.
\end{definition}

Consider a budget $B'$ which is not in the Smith set and another budget $B$ which is in the Smith set.
By definition, $B'$ is dominated by $B$,
and thus $B'$ does not weakly dominate $B$.
Further, there is no path from $B'$ to $B$,
since $B'$ does not dominate any budget in the Smith set,
and this also holds for all other budgets which are, like $B'$, not in the Smith set.
We conclude, by considering the counterpositive,
that in order to prove that SBA is Smith-consistent,
it is sufficient to show that any budget it computes has paths to all other budgets.

\begin{corollary}
Given a vote profile $\mathcal{V}$ over a budget proposal $\mathcal{P}$ and a budget limit $\ell$,
a budget $B$ is a Smith budget if and only if $B \dpath B'$ holds for any feasible budget $B'$ of $\mathcal{P}$.
\end{corollary}

A \emph{path} in a directed graph is a sequence of arcs, where the second vertex of an arc in the path is the first vertex of the next arc in the path. It will be useful to consider paths in the weak majority graph, referred to as \emph{weak majority paths}. We denote the existence of a weak majority path from one item $b$ to another $b'$ by $b \gpath b'$.
The lemmas below relate the weak majority graph to the budgets graph
(recall that the former is polynomial and has budget items as vertices, while the latter is exponential and has feasible budgets as vertices).

\begin{lemma}[Weak majority arc implies domination]\label{lemma:simplesimplecrucial}
  Let $b$ and $b'$ be two vertices in the majority graph computed by SBA for some profile $\mathcal{V}$ over some proposal $\mathcal{P}$. Let $B = \{b\}$ and $B' = \{b'\}$. If the arc $(b, b')$ is present in the weak majority graph then $B \dom B'$.
\end{lemma}

\begin{proof}
Towards a contradiction, assume that $B$ does not weakly dominate $B'$,
so $B'$ dominates $B$.
Let $M$ be a set of more than half of the voters where each voter $v \in M$ prefers $B'$ to $B$.
By definition, these voters rank $b'$ higher than they rank $b$.
Thus, the arc $(b', b)$ is present in the majority graph,
in contradiction to the assumption that $(b, b')$ is in the weak majority graph.
\end{proof}

\begin{lemma}[Weak majority path implies domination path]\label{lemma:simplecrucial}
  Let $b$ and $b'$ be two vertices in the majority graph computed by SBA for some vote profile $\mathcal{V}$ over some budget proposal $\mathcal{P}$. Let $B = \{b\}$ and $B' = \{b'\}$. If $b \gpath b'$ then $B \dpath B'$.
\end{lemma}

\begin{proof}
Since $b \gpath b'$ it follows that there is a set of vertices $b_1, \ldots, b_t$ such that $b = b_1$, $b_t = b'$, and the arc $(b_i, b_{i + 1})$ is present in the weak majority graph for each $i \in [t - 1]$. From Lemma~\ref{lemma:simplesimplecrucial} it follows that $B_{b_i} \dom B_{b_{i + 1}}$ holds for each $i \in [t - 1]$, where $B_{B_i} = \{b_i\}$ and $B_{B_{i + 1}} = \{b_{i + 1}\}$. We conclude that $B \dpath B'$.
\end{proof}

\begin{lemma}[Crucial lemma]\label{lemma:crucial}
  Let $\mathcal{P}$ be a budget proposal, $\ell$ be a budget limit, $B, B' \subseteq \mathcal{P}$ two feasible budgets, and $\mathcal{V}$ be a vote profile over $\mathcal{P}$.
Then, if $b \gpath b'$ for some $b \in B \setminus B'$ and $b' \in B' \setminus B$,
  then $B \dpath B'$.
\end{lemma}

\begin{proof}
Let $b \in B \setminus B'$ and $b' \in B' \setminus B$ such that $b \gpath b'$.
Let $B_b = \{b\}$ and notice that $B_b$ is feasible since $B_b \subseteq B$ and $B$ is feasible.
Further, $B \dpath B_b$ since in particular $B \dom B_b$ as $B_b \subseteq B$.
In order to show that $B \dpath B'$ it remains to show that $B_b \dpath B'$.
Since $b \gpath b'$ there is a set of vertices $b_1, \ldots, b_t$ such that $b = b_1$, $b_t = b'$,
and $(b_i, b_{i + 1})$ is in the weak majority graph for each $i \in [t - 1]$.
Notice that $b = b_1 \notin B'$ and that $b' = b_t \in B'$.
Let $j$ ($j \in [t - 1]$) be the smallest index for which $b_j \notin B'$  but $b_{j + 1} \in B'$.
(Indeed, it might be that $b' = b_{j + 1}$;
this would happen if the weak majority path from $b$ to $b'$ does not go through any vertex in $B'$.)
Using Lemma~\ref{lemma:simplecrucial} and since $b \gpath b_j$ it follows that $B_b \dpath B_j$ where $B_j = \{b_j\}$,
thus it remains to show that $B_j \dom B'$.
Using Lemma~\ref{lemma:simplesimplecrucial},
and since the arc $(b_j, b_{j+1})$ is in the weak majority graph,
it follows that $B_j \dom B_{j+1}$ where $B_{j+1} = \{j+1\}$.
Thus, for each set $M$ of more than half of the voters
there exists at least one voter $v \in M$ which ranks $b_j$ before it ranks $b_{j+1}$.
Since $b_j \notin B'$ but $b_{j+1} \in B'$,
it then follows that
$\min(\pos_v(B_j \setminus B')) = \min(\pos_v(B_j)) = \pos_v(B_j) < \pos_v(B_{j+1}) \leq \max(\pos_v(B' \setminus B_j))$.
As such a voter $v$ exists in every set $M$ containing more than half of the voters, it follows that $B_j \dom B'$.
\end{proof}

\begin{theorem}\label{theorem:mainnew}
  SBA is Smith-consistent.
\end{theorem}

\begin{proof}
Consider the ordered partition $V : C_1 \pref \dots \pref C_z$ produced by the Ranking procedure (each $C_j$ is a maximum subset of the $j$th Schwartz set).
Consider the budget $B$ produced by the Pruning procedure,
based on the ordered partition $V$ and let $B_i = B \cap C_i, i \in [z]$;
note that the $B_i$'s are exactly the maximal subsets selected for budgeting by the Pruning procedure.
Let $B'$ be a feasible budget and let $B'_i = B' \cap C_i, i \in [z]$.
Below we show that $B \dpath B'$ which will finish the proof.

If $B' \setminus B = \emptyset$ then $B' \subseteq B$ and thus $B'$ cannot dominate $B$, thus we assume that $B' \setminus B \neq \emptyset$.
Consider the smallest index $i$ ($i \in [z]$) for which $B'_i \setminus B_i \neq \emptyset$ and consider some $b' \in B'_i \setminus B_i$.
By construction, $b' \in C_i\setminus B_i$.  Since $B_i$ is a maximal subset of $C_i$ such that the budget $(B \cap \{C_1, \ldots, C_{i - 1}\}) \cup B_i$ is feasible
(refer to Line~\ref{line:greedy} in Algorithm~\ref{algorithm:dob}),
it follows that $(B \setminus B') \cap \{C_1, \ldots, C_i\} \neq \emptyset$.
Let $b \in (B \setminus B') \cap \{C_1, \ldots, C_i\}$,
let $B_b = \{b\}$,
and notice that $B_b$ is a feasible budget since $B_b \subseteq B$ and $B$ is feasible.
Further, notice that $B \dom B_b$ since $B_b \subseteq B$.
It remains to show that $B_b \dpath B'$ which would finish the proof.

Recall that $b' \in C_i$ and consider the following two cases.
The cases differ in whether $b$ is in the same component $C_i$ as $b'$ (Second case) or in an earlier component $C_j$, $j<i$ (First case).
In each case we show that $B_b \dpath B'$ and thus conclude that indeed $B \dpath B'$ and thus SBA is Smith-consistent.

\mypara{First case}
If $b \in C_j$ for some $j < i$,
then the arc $(b', b)$ cannot be in the majority graph,
since $C_j$ is a Schwartz set in the majority graph not containing the vertices in $C_1, \ldots, C_{j - 1}$ but containing $b'$.
Thus, the arc $(b, b')$ is present in the weak majority graph, thus $b \gpath b'$.
Applying Lemma~\ref{lemma:crucial} we conclude that $B_b \dpath B'$.

\mypara{Second case}
If $b \in C_i$,
then $b$ and $b'$ are in the same Schwartz set $C_i$.
If $b$ and $b'$ are in different Schwartz components, then the arcs $(b, b')$ and $(b', b)$ are not present in the majority graph,
thus they are present in the weak majority graph.
In particular, $b \gpath b'$, and by applying Lemma~\ref{lemma:crucial} we conclude that $B_b \dpath B'$.
Otherwise, if $b$ and $b'$ are in the same Schwartz component then $b \gpath b'$ since each Schwartz component is a cycle in the majority graph. Again, applying~\ref{lemma:crucial} we conclude that $B_b \dpath B'$.
\end{proof}

SBA is efficient,
as the majority graph has linearly-many vertices and computing a Schwartz set can be done in polynomial time (Schwartz components correspond to strongly connected components; this is folklore). Furthermore, choosing a maximal subset of a Schwartz set which is the closest to the previous budget can be done efficiently (e.g., greedily choosing items budgeted by $B_{-1}$).

\begin{corollary}
  SBA is a polynomial-time Smith-consistent budgeting algorithm.
\end{corollary}

\subsection{Reality-aware Budgeting and Hysteresis}\label{section:hysteresis}

Notice how SBA uses the previous budget $B_{-1}$ to break ties when selecting a maximal subset of each Schwartz set (in the pruning phase).
This is a special case of the Reality-aware voting rules considered in~\cite{realsoc} and is motivated, e.g., by Arrow's claim that “The status quo does have a built-in edge over all alternative proposals”~\cite[Page 95]{arrowbooksecond}.
More generally, considering the previous budget is both natural and useful, as budgeting scenarios fit very well in the framework of Reality-aware Social Choice~\cite{realsoc}:
  (1) Votes may be provided as a simple amendment to the previous budget, e.g., by adding/removing items or changing their quantities (indeed, recall that we can incorporate partial orders and also refer to the UI considerations at Section~\ref{section:outlook});
  and
  (2) The distance of a new budget to the previous budget can be computed, e.g., as the total cost of the items in the symmetric difference between the two (see the Distance-based Reality-aware Social Choice model~\cite{realsoc}).
   
Formally, the previous budget can be seen as a dichotomous ordered partition of the budget proposal,
where the first indifference class contains the previously-budgeted items and the second contains those which are not.
Then,
the previous budget can be used as follows:
  (1) Limit voters' flexibility, requiring them to submit preferences which are not too far from the previous budget;
  (2) Orthogonally, limit the distance of the new budget from the previous budget.
This might be a form of hysteresis (as discussed, e.g., in~\cite{shapiro2017foundations}),
limiting the rate of change of budgets year-to-year.

This perspective allows for the budget limit itself to be decided upon democratically, e.g., by taking the median of the proposed budget limits. This view of the previous budget as an ordered partition motivates the generalization described next.

\section{Incorporating Partial Orders}\label{section:partial orders}

In this section we adapt SBA to voters specifying partial orders. This would be especially useful for considering quantitative items (Section~\ref{section:quantitative items}).
The main observation is that we shall only care for the definition of when a voter \emph{prefers} some budget $B$ over another $B'$. For linear orders, we had Definition~\ref{def:prefers}. Next we consider ordered partitions.
Ordered partitions (linear orders with ties) are particularly interesting in the context of budgets,
as they can be understood as preference classes.

\begin{example}\label{example:power of ordered partitions}
Consider a budget proposal $\mathcal{P} = \{b_1, b_2, b_3, b_4, b_5\}$
and a vote $v : \{b_1\} \pref \{b_2, b_4\} \pref \{b_3, b_5\}$.
Such a vote can be understood as a voter wishing to first budget $b_1$;
then, if possible, budget $b_2$ and $b_4$, and only if the budget is not exhausted yet, budget $b_3$ and $b_5$.
\end{example}

The definition of when a voter prefers one budget to another extends as is from linear orders (i.e., Definition~\ref{def:prefers}) to ordered partitions. Notice that, in particular, it means that a voter is indifferent to two budgets having a nonempty symmetric difference in at last one component of her ordered partition.

\begin{example}
Let $v : \{a, b, c\} \pref \{d\}$ be an ordered partition over the budget proposal $\mathcal{P} = \{a, b, c, d\}$.
Let the cost of all items be $1$, and let $B = \{a, b\}$ and $B' = \{a, c\}$.
Then, $\pos_v(B \setminus B') = \pos_v(B' \setminus B) = \{1\}$,
thus $v$ neither prefer $B$ to $B'$ nor $B'$ to $B$.
\end{example}

It seems that ordered partitions provide enough practical expressive power for voters (see Example~\ref{example:power of ordered partitions}).
Still, the following definition is a generalization, wrt.\ to the minmax extension, for general partial orders.

\begin{definition}[Prefers for partial orders]\label{def:prefers for partial orders}
Let $v$ be a partial order over some budget proposal $\mathcal{P}$.
Then,
$v$ \emph{prefers} a budget $B \subseteq \mathcal{P}$ over a budget $B' \subseteq \mathcal{P}$
if the following hold:
  (1) For each item $b' \in B' \setminus B$ there is an item $b \in B \setminus B'$ which comes before $b'$ in the partial order;
  (2) For each item $b \in B \setminus B'$ there is no item $b' \in B' \setminus B$ which comes before $b$ in the partial order.
\end{definition}

Importantly, the proof of Theorem~\ref{theorem:mainnew} follows through with this definition, by adapting Lemma~\ref{lemma:crucial} to account for this adapted notion of preference among budgets.

\section{Quantitative Budgets}\label{section:quantitative items}

We extend our model to deal with quantitative budgets,
and show that SBA remains efficient even when quantities are given in binary encodings. To this end, we adapt the definition of a budget proposal to account for quantitative budgets.

\begin{definition}[Quantitative representation of budget proposals]
  In the \emph{quantitative representation} of budget proposals,
  a budget proposal is a set of tuples $(b_i, q_i)$,
  where $q_i$ is the \emph{quantity} of the budget item $b_i$.
  Each budget item $b_i$, having $q_i$ copies of it in the budget proposal,
  is associated with a \emph{cost function},
  where $F_{b_i}(q)$, $q \in [q_i]$, is the \emph{cost} of $q$ copies of $b_i$.
  The definition of cost is generalized to budgets in a natural way:
    The cost of a budget $B$ containing $q_i$ copies of budget item $b_i$, $i \in [z]$
    is
    $\sum_{j \in [z]} F_{b_i}(q_i)$.
\end{definition}

\begin{example}
Consider a proposal containing $5$ submarines, $3$ boats, and $1$ airplane.
Each submarine costs $7$ million, but there is a discount of $4$ million for each $5$ submarines;
each boat costs $2$ million; and an airplane costs $10$ million.
Then,
denoting submarines by $s$, boats by $b$, and airplanes by $a$,
the proposal is $\mathcal{P} = \{(s, 5), (b, 3), (a, 1)\}$ with cost functions:
  For submarines:
    $F_s(1) = 7$,
    $F_s(2) = 14$,
    $F_s(3) = 21$,
    $F_s(4) = 28$,
    $F_s(5) = 29$;
  For boats:
    $F_b(1) = 2$,
    $F_b(2) = 4$,
    $F_b(3) = 6$;
  For airplanes:
    $F_a(1) = 10$.
A voter can then specify:
``my first priority is 2 submarines; my second priority is 1 airplane and 3 boats; my last priority is another 3 submarines'';
by voting:
$v : \{(s, 2)\} \pref \{(a, 1), (b, 3)\} \pref \{(s, 3)\}$.
\end{example}

\begin{remark}
A voter, specifying an ordered partition over a quantitative budget proposal $\mathcal{P} = \{(b_i, q_i) : i \in [z]\}$,
can be represented as $v : C_1 \pref C_2 \pref \cdots \pref C_q$,
where each $C_j$ contains several tuples of the form $(b_j, z_{b_j})$, where $z_{b_j} \leq q_{b_j}$;
the sum of those quantities sums to $q_{b_j}$.
E.g., in the vote considered in the example above,
$v : \{(s, 2)\} \pref \{(a, 1), (b, 3)\} \pref \{(s, 3)\}$
corresponds to
$v : C_1 \pref C_2 \pref C_3$;
and there are, e.g., $5$ submarines in all of these components together.
\end{remark}

\subsection{Prefers with Quantities}

To adapt SBA to quantitative budget proposals,
first we adapt Definition~\ref{def:prefers} to define when does a voter prefer one budget over another.
We introduce the concepts of \emph{remainders} and \emph{ranked difference}. The remainder of a budget with respect to a vote consists of all the items ranked by a vote but not budgeted by the budget, keeping their original ranking in the vote.

\begin{definition}[Remainder]
Let $\mathcal{P}$ be a quantitative budget proposal,
$B$ a budget of $\mathcal{P}$,
and $v : C_1 \pref C_2 \pref \dots \pref C_z$ an ordered partition over $\mathcal{P}$.
The unbudgeted \emph{remainder} of $B$ with respect to $v$ is 
$Rem_B(v) = C'_1 \pref C'_2 \pref \dots \pref C'_z$, with components $C'_i$, $i \in [z]$, defined via the following equations:
    $B_0 = B$, and for each $i$, $i \in [z]$,
    $C'_i = C_i \setminus B_{i - 1}$
    and $B_i = B_{i - 1} \setminus C_i$.
\end{definition}

Intuitively,
the $i^{th}$ component of the remainder of a budget $B$ of $\mathcal{P}$ with respect to $v$
contains exactly the elements of the $i^{th}$ component of $v$ which are unbudgeted by $B$.
Given a vote $v$ and a budget $B$ over $\mathcal{P}$,
the union of the elements of $Rem_B(v)$ equals $\mathcal{P} \setminus B$,
which are all items unbudgeted by $B$.
We use remainders to compute the ranking of the elements in the symmetric difference between two budgets.

\begin{definition}[Ranked difference between budgets]
Let $\mathcal{P}$ be a quantitative budget proposal,
$v$ a ranking over $\mathcal{P}$,
and $B$ and $B'$ budgets over $\mathcal{P}$.
The \emph{ranked difference}
between the unbudgeted remainders of the two budgets
with respect to $v$ is $\Diff_{B,B'}(v) = Rem_{B'}(v) \setminus Rem_{B}(v)$, where the set difference operator is applied component-wise.
We define $\Ind_{B,B'}(v)$ as the set of indices of the components for which $\Diff_{B,B'}(v)$ is non-empty;
that is, denoting $\Diff_{B,B'}(v)$ as $C_1 \pref \dots \pref C_z$, we define $\Ind_{B,B'} = \{i : C_i \neq \emptyset\}$.
\end{definition}

That is,
while $Rem_B(v)$ collects and ranks items unbudgeted by $B$ according to their ranking by $v$,
$\Diff_{B,B'}(v)$ collects and ranks items budgeted by $B$ but not by $B'$ according to their ranking,
and each index in $\Ind_{B,B'}(v)$ names a component of $v$ that has at least one item budgeted by $B$ but not by $B'$.

Next, intuitively, a vote prefers one budget over another if it ranks higher all the items budgeted by the first budget but not the second, compared to all items budgeted by the second budget but not by the first.

\begin{definition}[Prefers for quantitative budgets]\label{def:preferstwo}
Let $\mathcal{P}$ be a proposal,
$v$ a ranking over $\mathcal{P}$,
and $B$ and $B'$ budgets over $\mathcal{P}$.
We say that \emph{$v$ prefers $B$ over $B'$} if it holds that
$\max(\Ind_{B,B'}(v)) < \min(\Ind_{B',B}(v))$,
where we set $max(\emptyset) = \min(\emptyset) = \infty$.
\end{definition}

E.g., if $B' \subset B$,
then $B$ is preferred over $B'$ (by any vote $v$) since $\Ind_{B',B}(v) = \emptyset$ while $\Ind_{B,B'}(v)$ is not.
In addition, \emph{prefers} is irreflexive since $\infty \nless \infty$.
This definition indeed generalizes the simple definition of the basic model, described in Section~\ref{section:basic model}.
To illustrate the notions defined above, consider the following example.

\begin{example}
Let $\mathcal{P} = \{(a, 3), (b, 1)\}$ be a quantitative budget proposal,
where the cost of each item is $1$.
Let $\mathcal{V} = \{v\}$ be a profile,
where $v : \{(a, 1)\} \pref \{(a, 1)\} \pref \{(b, 1)\} \pref \{(a, 1)\}$
(each occurrence of $(a, 1)$ is a different `$a$`-item).
Let $B = \{\{(a, 2)\}, \{b, 1\}\}$ and $B' = \{\{(a, 3)\}\}$ be two budgets.
Then,
$Rem_{B}(v) = \{\} \pref \{\} \pref \{\} \pref \{a\}$,
$Rem_{B'}(v) = \{\} \pref \{\} \pref \{b\} \pref \{\}$,
$\Ind_{B,B'}(v) = \{3\}$,
and 
$\Ind_{B',B}(v) = \{4\}$.
Thus, $v$ prefers $B$ over $B'$.
\end{example}

\subsection{SBA is Efficient for Quantitative Proposals}

A straightforward adaptation of SBA to quantitative budget proposals
would be to rename the items such that the quantity of each item would be one,
and then to use the description of SBA from Section~\ref{section:basic model}; this, however, results in pseudo-polynomial time.

It is possible to slightly modify SBA to run in polynomial time for
quantitative budget proposals.
To explain our modifications to SBA,
it is useful to identify the cause of the computational inefficiency of the straightforward adaptation of SBA:
  The constructed majority graph would have a vertex for each single copy of each item, possibly resulting in a super-polynomial-sized majority graph. To remedy this inefficiency, it seems natural to construct a majority graph with one vertex for each item,
representing all copies of that item:
  The resulting majority graph would indeed have polynomial number of vertices. However, we would not be able to represent all majority relations between the budget items, as such graph would be, roughly speaking, too ``coarse''.
To overcome this difficulty,
we will first initiate the majority graph as described above,
followed by splitting its vertices according to the way the corresponding copies of each item are split in the vote profile.

\begin{theorem}\label{theorem:poly}
 SBA can be adapted to quantitative budget proposals to run in polynomial time.
\end{theorem}

\begin{proof}
We describe a modification to SBA,  denoted by ESBA, which differs from SBA in the way it constructs the majority graph.
Specifically, first construct a majority graph $G$ having a vertex $(b, q)$ for each tuple $(b, q)$ in the quantitative budget proposal.
Then, iteratively process each tuple $(b, q)$ in the quantitative budget proposal, initially corresponding to a vertex $(b, q)$ in $G$, as follows. For each vote $v : C_1 \pref \dots \pref C_z$,
represented over the quantitative representation of budget proposals, and where $b$ appears in the components $\mathcal{C} = \{C_{i_1}, C_{i_2}, \ldots, C_{i_t}\}$ ($i_j \in [z]$, $j \in [t]$)
with corresponding quantities $q_{i_1}, q_{i_2}, \ldots, q_{i_t}$,
such that $(b, q_{i_j}) \in C_{i_j}$ ($i_j \in [z]$, $j \in [t]$),
notice that $\sum_{j \in [t]} q_{i_j} = q$, and create the following sequence of numbers, denoted by $\splittt_v(b)$:
  $[q_{i_1}, q_{i_1} + q_{i_2}, \ldots, \sum_{j \in [t]} q_{i_j}]$.
Intuitively, $\splittt_v(b)$ collects $v$'s \emph{splitting points} of $b$. E.g., consider a vote $v : (a, 2) \pref (b, 2) \pref (a, 1)$:
  Notice that $v$ splits the $3$ occurrences of $a$ into two components, where in the first component $v$ has two of its occurrences while in the second component $v$ has the third occurrence of $b$. Correspondingly, $\splittt_{v_3}(a) = [2, 3]$.

After computing $\splittt_v(b)$ for each vote $v$,
merge those sequences into a single sequence denoted by $\splittt(b)$;
  that is, $\splittt(b) = \cup_{v \in \mathcal{V}} \splittt_v(b)$,
  where $\mathcal{V}$ is the given vote profile.
Intuitively, $\splittt(b)$ collects the \emph{splitting points} of $b$ across all voters.
Recall that currently the majority graph has one vertex $(b, q)$ corresponding to the budget item $b$;
first, we delete $(b, q)$ from the majority graph.
Then,
denote $\splittt(b) = [q_1, \ldots, q_x]$
and create the following $x$ vertices
$\{(b, [1, q_1]), (b, [q_1 + 1, q_2]), \ldots, (b, [q_{x - 1} + 1, q_x])\}$.
Intuitively, a vertex $(b, [l, r])$ stands for the $l^{th}$, $(l + 1)^{th}$, $\ldots$, $(r - 1)^{th}$, and $r^{th}$ occurrence of $b$.

The above process is done for each tuple $(b, q)$ in the quantitative budget proposal. Importantly, the number of vertices in the resulting majority graph is polynomially bounded by the input,
as any splitting point in $\splittt(b)$, and thus each vertex in the majority graph, can be charged to at least one vote,
and each vote can be responsible only for polynomially-many vertices. Furthermore, the constructed majority graph can represent all majority relations, as no vote splits a budget item $b$ ``finer'' than $\splittt(b)$.
\end{proof}

\section{Hierarchical Budgeting}\label{section:hierarchy}

Budgets of complex organizations and societies are constructed hierarchically; here we use our approach, specifically the fact that SBA works in two independent phases, to achieve that.

Hierarchical budgets are constructed in two phases.  First, each section of the organization performs an independent budgeting process that decides upon the priorities within the section, but without a given budget limit. Then,
a consolidating budgeting process determines the allocation of funds to the different sections by providing each with a \emph{section budget limit}, the sum of which is the consolidated budget limit.  Each section's ordered partition is pruned by its section budget limit.  The resulting section budgets are combined into a consolidated feasible budget.
E.g., consider a government with several ministries:
  each ministry conducts its own budgeting process and then the government conducts its overall budgeting process by allocating resources to each of the ministries.  

An important feature of our algorithm is that it does not require fixing an a priori section budget limits, as the ranking procedure does not depend on the limit.  This flexibility turns out to be crucial for the algorithmic construction of hierarchical budgets. 
It also has important practical implications:
  The allocation of section budget limits need not be done a priory, but can use the ranking of the section to take into account and deliberate the implications  of different section budget limits on the budgeting of various items.  So while the government might not intervene in the ranking of the section (perhaps respecting the subsidiary principle), it has some control of what will be funded and what will not by choosing a specific section budget limit.

Our hierarchical budgeting method can be applied in two scenarios:
One, where each section budget is created by votes of experts/stakeholders, and then the consolidated budget is created by the votes of the sovereign body. Another, when the same body decides on both section budgets and the consolidated budget. This method helps the sovereign body to focus in turn on each section budget independently, and then to consolidate the results.

Our next description has only two levels of hierarchy,
however the process naturally generalizes to further levels.
Given $n$ sections, we begin by creating $n$ budget proposals $\mathcal{P}_i, i \in [n]$, one for each section.
Then, the voters of each section vote independently on each section proposal. We separately apply only the ranking procedure to the votes of each section to produce the section rankings $V_i, i \in [n]$. We arbitrarily linearize the unary rankings; denote the resulting sequences $V^L_i, i \in [n]$ and define $V^L_i[k]$ to be the set of the first $k$ items of $V^L_i$.
We then define $n$ ``artificial'' derived budget items $s_i, i \in [n]$, one for each section,
and associated cost functions $F_i(n) = \cost(V^L_i[n]-\cost(V^L_i[n-1])$.
The consolidated budget proposal is then $\mathcal{P} = \bigcup_{i\in [n]} (s_i, n_i)$, where $n_i$ is the number of of items in $V^L_i$.  
Then, each vote on the consolidated budget proposal is a ranking of the consolidated proposal $\mathcal{P}$, which consists of only $s_i$'s. (Practically, however, when a voter chooses, say, to rank first six $s_3$'s, this would be after the voter has looked ``under the hood'' and knows what are the first six items of the budget of the $s_3$ section.) Given a budget limit $\ell$ on the consolidated budget, all votes on the consolidated budget proposal $\mathcal{P}$ are then treated as usual:
  We apply the ranking procedure and then the pruning procedure with the limit~$\ell$, and produce a high-level budget.

With respect to the allocation between sections, the properties of this consolidated budget are the same as the properties of the non-hierarchical budgets produced by our algorithm. Votes on the consolidated proposal can prioritize items among the different sections, and know what are they prioritizing, but cannot affect,
at the consolidation stage, the relative priority of the underlying ``true'' items within each section.

To compare constructing a budget to constructing it hierarchically, we assume that the voters have preferences over all items, of all sections, and that, to ease the budgeting process, they decide to construct it hierarchically.
Then,
we compare the two options:
  (1) Using SBA directly on the preferences of the voters;
  and
  (2) Using SBA hierarchically,
      first on each section independently,
	  and then consolidating the budget at a high level, over the sections, as described above.
The resulting budget might be different.

\begin{example}
Consider two sections, Section A, containing items $a_1$ and $a_2$, and Section B, containing item $b$.
Let all items cost $1$, and set the budget limit to~$2$.
Assume the following voters with preferences over all items of all sections together:
  $v_1 : a_1 \pref a_2 \pref b$,
  $v_2 : a_1 \pref b \pref a_2$,
  $v_3 : a_2 \pref b \pref a_1$.
The budget $\{a_1, a_2\}$ is a Condorcet-winning budget
(notice that both $v_1$ and $v_2$ prefer it to $\{a_2, b\}$ while both $v_1$ and $v_3$ prefer it to $\{a_1, b\}$);
thus, using SBA directly, disregarding hierarchy, would result in $\{a_1, a_2\}$ being the winning budget,
while using SBA hierarchically would result in a different winning budget:
  Using the ranking phase of SBA on Section A would give the ordered partition $a_1 \pref a_2$, and using the ranking phase of SBA on Section B would give the ordered partition $b$. Then, considering those ordered partitions and employing SBA in the high-level,
would result in the budget $\{a_1, b\}$.
\end{example}

There are certain cases in which using SBA directly (disregarding hierarchy) and using it hierarchically is guaranteed to give the same results. One simple, extreme such case is where all voters agree unanimously on the order of the low-level section items of all sections.

\section{Discussion}\label{section:outlook}

We identified certain aspects, generally disregarded in the growing literature on PB,
which are nevertheless useful in employing such methods to produce budgets in a democratic way. Accordingly, we provided a formal treatment of a general setting of democratic budgeting which
(a) Takes into account the previous (last year's) budget;
(b) Allows voters to specify partial orders or special cases of such; (c) Incorporates quantitative, indivisible items of arbitrary costs, capturing, e.g., items with decreasing marginal costs; and (d) Allows for hierarchical construction of budgets.
We generalized the Condorcet criterion to democratic budgeting
and described a polynomial-time Condorcet-consistent participatory budgeting. Our algorithm is designed in such a way that it supports hierarchical budget construction, making it applicable to entire, high-stakes budgets.
Next, we discuss avenues for further research.

\mypara{Complementarities}
Our elicitation method allows for quite complex user preferences, but does not incorporate complementarities between budget items.  Such inter-relations between budget items have not been considered in the literature on PB. Having efficient ways for users to express such, while avoiding exponential blow-up, is an interesting future research direction.

\mypara{More Axiomatic Considerations}
Here we concentrated on Condorcet-consistency. This makes our methods suitable for situations where majority decisions are morally justified and not to situations where proportionality is desired (where one can use, e.g., the methods of Aziz et al.~\cite{aussieone}; notice that usually majoritarianism and proportionality are somewhat mutually exclusive). On a different note, one might consider other axiomatic properties, such as monotonicity. 

\mypara{Aspects of Democratic Budgets}
One might consider other methods of democratic budgeting,
e.g., investigating how to generalize existing methods of participatory budgeting to be applied to the more general setting of democratic budgeting.

\mypara{Breaking Ties}
SBA picks a maximal subset of each Schwartz set which is the closest to the previous budget. We suggest further to (1) Sort items by cost, favoring cheaper items, resulting in budgets with more items; and (2) Sort items by indices, favoring an item which was never budgeted to another item of a type which was already budgeted several times,
resulting in more diverse budgets.
One might also aim at solving the underlying Knapsack problem,
by computing (e.g., using dynamic programming) the maximum-sized subset of each Schwartz set.

\mypara{User Interface Considerations}
For deployment, and if a vote is a weak order, a convenient UI is needed to specify it. Two simple user interfaces are a sequence of pairs (item, quantity) and a pie chart. We interpret the former as an ordered partition where each partition consists of one pair. A pie chart might be interpreted as one component with multiple pairs, but we propose to interpret it differently, as follows:
  If not all items can be funded, then items should be funded proportionally according to the cost ratios specified by the chart. As our budgeting model refers to discrete items, each with a specific cost, such proportional allocation can only be approximated. To exploit the full power of ordered partitions, a user interface should cater for a sequence of pie charts.
Furthermore,
the user interface may directly incorporate the previous budget,
viewed as a dichotomous ordered partition,
and allow voters to move items from the two indifference classes,
and also to provide for the possibility of refining them to more classes.

\bibliographystyle{ieeetr}
\bibliography{bib}

\begin{thebibliography}{10}

\bibitem{cabannes2004participatory}
Y.~Cabannes, ``Participatory budgeting: a significant contribution to
  participatory democracy,'' {\em Environment and urbanization}, vol.~16,
  no.~1, pp.~27--46, 2004.

\bibitem{wainwright2003making}
H.~Wainwright, ``Making a people's budget in {P}orto {A}legre,'' {\em NACLA
  Report on the Americas}, vol.~36, no.~5, pp.~37--42, 2003.

\bibitem{ganuza2012power}
E.~Ganuza and G.~Baiocchi, ``The power of ambiguity: How participatory
  budgeting travels the globe,'' {\em Journal of Public Deliberation}, vol.~8,
  no.~2, 2012.

\bibitem{kilgour2010approval}
D.~M. Kilgour, ``Approval balloting for multi-winner elections,'' in {\em
  Handbook on approval voting}, pp.~105--124, 2010.

\bibitem{realsoc}
E.~Shapiro and N.~Talmon, ``Incorporating reality into social choice,'' in {\em
  Proceedings of the 17th Conference on Autonomous Agents and Multiagent
  Systems (AAMAS '18)}, 2017.

\bibitem{aussieone}
H.~Aziz, B.~Lee, and N.~Talmon, ``Proportionally representative participatory
  budgeting: Axioms and algorithms,'' in {\em Proceedings of the 17th
  Conference on Autonomous Agents and Multiagent Systems (AAMAS' 18)}, 2017.

\bibitem{goel2015knapsack}
A.~Goel, A.~K. Krishnaswamy, S.~Sakshuwong, and T.~Aitamurto, ``Knapsack
  voting,'' {\em Collective Intelligence}, 2015.

\bibitem{benade2017preference}
G.~Benade, S.~Nath, A.~D. Procaccia, and N.~Shah, ``Preference elicitation for
  participatory budgeting,'' in {\em Proceedings of the 31st AAAI Conference on
  Artificial Intelligence (AAAI '17)}, pp.~376--382, 2017.

\bibitem{moulin2016handbook}
H.~Moulin, F.~Brandt, V.~Conitzer, U.~Endriss, A.~D. Procaccia, and J.~Lang,
  {\em Handbook of Computational Social Choice}.
\newblock 2016.

\bibitem{chaptermw}
P.~Faliszewski, P.~Skowron, A.~Slinko, and N.~Talmon, ``Multiwinner voting: A
  new challenge for social choice theory,'' in {\em Trends in Computational
  Social Choice} (U.~Endriss, ed.), 2017.

\bibitem{csrhierarchy}
P.~Faliszewski, P.~Skowron, A.~Slinko, and N.~Talmon, ``Committee scoring
  rules: {A}xiomatic classification and hierarchy,'' in {\em Proceedings of the
  27nd Joint Conference on Artifical Intelligence (IJCAI '16)}, pp.~250--256,
  2016.

\bibitem{xia2011determining}
L.~Xia and V.~Conitzer, ``Determining possible and necessary winners given
  partial orders,'' {\em Journal of Artificial Intelligence Research}, vol.~41,
  pp.~25--67, 2011.

\bibitem{lu2011budgeted}
T.~Lu and C.~Boutilier, ``Budgeted social choice: From consensus to
  personalized decision making,'' in {\em Proceedings of the 22nd Joint
  Conference on Artifical Intelligence (IJCAI '11)}, pp.~280--286, 2011.

\bibitem{fishburn1981majority}
P.~C. Fishburn, ``Majority committees,'' {\em Journal of Economic Theory},
  vol.~25, no.~2, pp.~255--268, 1981.

\bibitem{fishburn1981analysis}
P.~C. Fishburn, ``An analysis of simple voting systems for electing
  committees,'' {\em J. on Appl. Math.}, vol.~41, no.~3, pp.~499--502, 1981.

\bibitem{darmann2013hard}
A.~Darmann, ``How hard is it to tell which is a condorcet committee?,'' {\em
  Mathematical social sciences}, vol.~66, no.~3, pp.~282--292, 2013.

\bibitem{sekar2017condorcet}
S.~Sekar, S.~Sikdar, and L.~Xia, ``Condorcet consistent bundling with social
  choice,'' in {\em Proceedings of the 16th Conference on Autonomous Agents and
  Multiagent Systems (AAMAS' 17)}, pp.~33--41, 2017.

\bibitem{kamwa2017stable}
E.~Kamwa, ``Stable rules for electing committees and divergence on outcomes,''
  {\em Group Decision and Negotiation}, vol.~26, no.~3, pp.~547--564, 2017.

\bibitem{aziz2017condorcet}
H.~Aziz, E.~Elkind, P.~Faliszewski, M.~Lackner, and P.~Skowron, ``The condorcet
  principle for multiwinner elections: From shortlisting to proportionality,''
  {\em arXiv preprint arXiv:1701.08023}, 2017.

\bibitem{barbera2004ranking}
P.~K.~P. S.~Barbera, W.~Bossert, ``Ranking sets of objects,'' in {\em Handbook
  of Utility Theory} (C.~S. S.~Barbera, P. J.~Hammond, ed.), vol.~II, ch.~17,
  pp.~893--977, 2004.

\bibitem{aziz2016computing}
H.~Aziz, J.~Lang, and J.~Monnot, ``Computing pareto optimal committees,'' in
  {\em Proceedings of the 25th International Joint Conference on Artificial
  Intelligence (IJCAI '16)}, pp.~60--66, 2016.

\bibitem{Bouy04a}
D.~Bouyssou, ``Monotonicity of `ranking by choosing': {A} progress report,''
  {\em Social Choice and Welfare}, vol.~23, no.~2, pp.~249--273, 2004.

\bibitem{arrowbooksecond}
K.~J. Arrow, {\em Social choice and individual values}.
\newblock Yale university press, second~ed., 1963.

\bibitem{shapiro2017foundations}
E.~Shapiro, ``Foundations of e-democracy,'' {\em arXiv preprint
  arXiv:1710.02873}, 2017.
\newblock To appear in Communications of the ACM, August 2019.

\end{thebibliography}

\end{document}